\documentclass{llncs}

\usepackage[english]{babel}
\usepackage{url}
\usepackage{enumerate}
\usepackage{wrapfig}
\usepackage{bbm}
\usepackage{flushend}
\usepackage[capitalise]{cleveref}
\usepackage{color}
\usepackage{float}
\usepackage[caption=false]{subfig}
\usepackage{graphicx}
\usepackage{amsfonts}
\usepackage{amssymb}
\usepackage{multirow}
\usepackage{longtable}
\usepackage[colorinlistoftodos]{todonotes}

\usepackage{calc}
\renewcommand\thetable{\arabic{table}}

\newtheorem{thm}{Theorem}
\newtheorem{cor}[thm]{Corollary}
\newtheorem{lem}[thm]{Lemma}

\newcommand{\prob}[1]{\mathbb{P}\left(#1\right)}
\normalsize

\begin{document}

\title{Modelling trend progression through an extension of the Polya Urn Process}
\titlerunning{Trend modelling using Polya Urn Process} 

\author{Marijn ten Thij \and Sandjai Bhulai}

\authorrunning{M. ten Thij \and S. Bhulai} 

\tocauthor{Marijn ten Thij and Sandjai Bhulai}

\institute{Vrije Universiteit Amsterdam, Faculty of Sciences, the Netherlands,\\
\email{\{m.c.ten.thij,s.bhulai\}@vu.nl}
}
\maketitle

\begin{abstract}
Knowing how and when trends are formed is a frequently visited research goal. In our work, we focus on the progression of trends through (social) networks. We use a random graph (RG) model to mimic the progression of a trend through the network. The context of the trend is not included in our model. We show that every state of the RG model maps to a state of the Polya process. We find that the limit of the component size distribution of the RG model shows power-law behaviour. These results are also supported by simulations.
\end{abstract}


\section{Introduction}\label{sec:intro}
How can we reach a large audience with our message? What is the best way to reach a large audience with an advertisement? These are questions that are asked many times in our modern day society. Not only for corporate interest, but also for public interest by governments and charities. Everyone wants to get their message across to a large audience. Finding out how to do this is a frequently visited research goal in many fields, e.g. economics \cite{simon1955class}, evolutionary biology \cite{yule1925mathematical,ewens2012mathematical} and physics \cite{redner1998popular}.

In our work, we focus on the progression of trends through the network of users. Incident to our approach, we focus on the microscopic dynamics of user-to-user interaction to derive the overall behaviour, which is similar to the approach used in other works, e.g. \cite{watts2002simple,gleeson2014simple}. Our work differs from these in that we model the spread of the messages in a step-by-step fashion, whereas \cite{watts2002simple,gleeson2014simple} use a given degree distribution per user as a start of their analysis.

Our goal is to devise a model that mimics the progression of a trend through (social) networks. By doing this, we focus only on the pattern of progression of the trends and not their content. Based on observations from \emph{Twitter} data, we have built a model that captures the different changes that occur in a network whilst a topic is spreading. In \cite{thij2014modelling}, we derived basic growth properties of the model and the speed of convergence of these properties. In this paper, we derive the component size distribution of the model.

In Section~\ref{sec:related}, we describe related fields of study, and in Section~\ref{sec:model} we introduce the RG model and the Polya process. Here we show that the RG model can be easily mapped to the Polya process. Then, in Section~\ref{sec:theory}, we derive the behaviour of the component size distribution. Finally, we state our conclusions and discuss the possibilities for further research in Section~\ref{sec:conclusion}.

\section{Related work}\label{sec:related}
There are many studies that focus on information diffusion in online social networks. In \cite{Guille2013}, Guille et al. provide a survey, in which they distinguish two approaches: a descriptive approach and a predictive approach. A few examples of these approaches are given below.

A large number of descriptive studies into information diffusion have added to the knowledge about how messages progress through online social networks. Lerman and Ghosh \cite{lerman2010information} find that diffusion of information in \emph{Twitter} and \emph{Digg} is very dependent on the network structure. Bhattacharya and Ram \cite{bhattacharya2012sharing} study the diffusion of news items in \emph{Twitter} for several well-known news media and find that these cascades follow a star-like structure. Friggeri et al. \cite{friggeri2014rumor} study rumors in Facebook and find bursty spreading patterns. Sadikov and Martinez \cite{sadikov2009information} find that on \emph{Twitter}, tags tend to travel to distant parts of the network and URLs tend to travel shorter distances. Bhamidi et al. \cite{bhamidi2012twitter} propose and validate a so-called superstar random graph model for a giant component of a retweet graph. Hoang and Lim \cite{hoang2012virality} propose a dynamic model that uses both user and item virality and user susceptability as descriptors to model information diffusion. Iwata et al. \cite{iwata2013discovering} use an inhomogeneous Poisson Process to model the diffusion of item adoption. Another angle to model information diffusion uses epidemic spreading: By a maximum entropy argument Bauckhage et al. \cite{bauckhage2015parameterizing} find a closed form expression for the path length distribution in a network. Finally, Carton et al. \cite{carton2015audience} propose to perform an audience analysis when analysing diffusion processes, thus including not only the diffusers in the analysis but also the receivers of the content.

Romero et al. \cite{romero2011differences} use the predicitive approach to analyse the spread mechanics of content through hashtag use and derive probabilities that users adopt a hashtag. Kupavskii et al. \cite{kupavskii2012prediction} predict the number of retweets based on several features. They find that the flow of a message is one of the most important features in the prediction. Altshuler et al. \cite{altshuler2012trends} use past information and diffusion models to derive a lower bound on the probability that a topic becomes trending. Zaman et al. \cite{zaman2010predicting} predict future retweets based on features at the user level. Wu and Raschid \cite{wu2014prediction} define a user specific potential function which reflects the likelihood of a follower sharing the users content  in the future.

Classification and clustering of trends on \emph{Twitter} has also attracted considerable attention in the literature. Zubiaga et al. \cite{zubiaga2011classifying} derive four different types of trends, using fifteen features to make their distinction. They distinguish trends triggered by news, current events, memes or commemorative tweets. Lehmann et al. \cite{lehmann2012dynamical} study different patterns of hashtag trends in \emph{Twitter}. They also observe four different classes of hashtag trends. Rattanaritnont et al. \cite{rattanaritnontstudy} propose to distinguish topics based on four factors, namely cascade ratio, tweet ratio, time of tweet and patterns in topic-sensitive hashtags. Ferrara et al. \cite{Ferrara2013} cluster memes based on content and user, tweet and network similarity measures.

We use the analysis of urn processes in this paper, in contrast to other works in this area. Pemantle \cite{Pemantle2006} presents a survey of different techniques that are used in this field of research. In this work, we focus on extensions of the Polya urn problem, which is thoroughly analysed by Chung et al. in \cite{chung2003generalizations}. Specifically, we are interested in the infinite generalized Polya urn model, as studied in Section 4 of \cite{chung2003generalizations}. 

\section{Problem formulation}\label{sec:model} 
In this section, we first describe the setup of the RG model. Then, the Polya process is introduced. Finally, we show that every state of the model maps to a state of the Polya process.

Our main object of study is the retweet graph $G=(V,E)$, which is a graph of users that have participated in the discussion of a specific topic. A directed edge $e=(u,v)$ indicates that user $v$ has retweeted  a tweet of $u$. We observe the retweet graph at the time instances $t=0,1,2,\ldots$, where either a new node or a new edge is added to the graph, and we denote by $G_t=(V_t,E_t)$ the retweet graph at time $t$. As usual, the out- (in-) degree of node $u$ is the number of directed edges with source (destination) in $u$. For every new message initiated by a new user $u$ a tree $H_u$ is formed. Then, ${\cal T}_t$ denotes the forest of message trees. Note that in the RG model, a new message from an already existing user $u$ (that is, $u\in {\cal T}_t$) does not initiate a new message tree. We define $|{\cal T}_t|$ as the number of new users that have started a message tree up to time $t$. I.e. $G_t$ can be seen as an simple representation of the union of message trees $\cup_{H_u\in{\cal T}_t}H_u$.

The goal of the model is to capture the development of trending behaviour. We do this by combining the spread of several messages. As a result of this approach, we first need to model the progression of a single message in the network. To this end, we use the superstar model of Bhamidi et al. \cite{bhamidi2012twitter} for modelling distinct components of the retweet graph, to which we add the mechanism for new components to arrive and the existing  components to merge. In this paper, our aim is to analyse the component size distribution of  $G_t$. For the sake of simplicity of the model, we neglect the friend-follower network of \emph{Twitter}. Note that in \emph{Twitter} every user can retweet any message sent by any public user, which supports our simplification. 

We consider the evolution of the retweet graph in time $(G_t)_{t\ge 0}$.  We use a subscript $t$ to indicate $G_t$ and related notions at time $t$. We omit the index $t$ when referring to the graph at $t\rightarrow\infty$. Let $G_0$ denote the graph at the start of the progression. In the analysis of this paper, we assume $G_0$ consists of a single node. Note that in reality, this does not need to be the case: any directed graph can be used as an input graph $G_0$. 

Recall that $G_t$ is a graph of {\it users}, and an edge $(u,v)$ means that $v$ has retweeted a tweet of $u$. We consider  time instances $t=1,2,\ldots$ when either a new node or a new edge is added to the graph $G_{t-1}$. We distinguish three types of changes in the retweet graph:

\begin{itemize}
\item[$\circ$] $T_1$: a new user $u$ has posted a new message on the topic, so node $u$ is added to $G_{t-1}$;
\item[$\circ$] $T_2$: a new user $v$ has retweeted an existing user $u$,  so node $v$ and edge $(u,v)$ are added to $G_{t-1}$;
\item[$\circ$] $T_3$: an existing user $v$ has retweeted another existing user $u$, so edge $(u,v)$ is added to $G_{t-1}$.
\end{itemize}

Note that the initial node in $G_0$ is equivalent to a $T_1$ arrival at time $t=0$. Assume that each change in $G_t$ at $t=1,2,\ldots$ is $T_1$ with probability $\lambda/(1+\lambda)$, independently of the past. Also, assume that a new edge (retweet) is coming from a new user with probability $p$. Then, the probabilities of $T_1$, $T_2$ and $T_3$ arrivals are, $\frac{\lambda}{\lambda+1}$, $\frac{p}{\lambda+1}$, $\frac{1-p}{\lambda+1}$ respectively. The parameter $p$ governs the process of components merging together, while $\lambda$ governs the arrival of new components in the graph. 

For both $T_2$ and $T_3$ arrivals we define the same mechanism to choose the source of the new edge $(u,v)$ as follows. 

Let $u_0,u_1,\ldots$ be the users that have been added to the graph as $T_1$ arrivals, where $u_0$ is the initial node. Denote by $H_{i,t}$ the subgraph of $G_t$ that includes $u_i$ and all users that have retweeted the message of $u_i$ in the interval $(0,t]$. We call such a subgraph a message tree with root $u_i$.  We assume that the probability that a $T_2$ or $T_3$ arrival at time $t$ will attach an edge to one of the nodes in $H_{i,t-1}$ with probability $p_{H_{i,t-1}}$ is proportional to the size of the message tree:
\[ p_{H_{i,t-1}} = \frac{|H_{i,t-1}|}{\sum_{H_{j,t-1}\subset{\cal T}_{t-1}}|H_{j,t-1}|}. \]
This creates a preferential attachment mechanism in the formation of the message trees. 

For the selection of the source node, we use the superstar model, with parameter $q$ chosen uniformly for all message trees. This model was suggested in \cite{bhamidi2012twitter} for modelling the largest connected component of the retweet graph on a given topic, in order to describe a progression mechanism for a single retweet tree. Our extensions compared to \cite{bhamidi2012twitter} are that we allow new message trees to appear ($T_1$ arrivals), and that different message trees may either remain disconnected or get connected by a $T_3$ arrival.  

For a $T_3$ arrival, the target of the new edge $(u,v)$ is chosen uniformly at random from $V_{t-1}$, with the exception of the earlier chosen source node $u$, to prevent self-loops. That is, any user is equally likely to retweet a message from another existing user. Thus, after a $T_3$ arrival a message tree can have cycles.

Note that we do not include tweets and retweets that do not result in new nodes or edges in the retweet graph. This could be done, for example, by introducing dynamic weights of vertices and edges, that increase with new tweets and retweets. Here, we consider only an unweighted model.
\subsection*{Polya process}
In our analysis of the previously stated model, we use the Polya process, which is defined in \cite{chung2003generalizations} as follows:
\begin{center}\parbox{0.8\textwidth}{Given two parameters $\gamma\in\mathbb{R}, 0\leq \bar{p} \leq 1$, we start with one bin, containing one ball. We then introduce balls one at a time. For each new ball, with probability $\bar{p}$, we create a new bin and place the ball in that bin; with probability $1-\bar{p}$, we place the ball in an existing bin (of size $m$), such that the probability that the ball is placed in a bin is proportional to $m^\gamma$.}\end{center}
We only consider the case where $\gamma=1$ in this paper. 

Let $f_{i,t}$ denote the fraction of bins that contain $i$ balls at time $t$. In \cite{chung2003generalizations}, the authors find that under the following assumptions:
\begin{enumerate}[i)]
\item for each $i$, there exists $f_i\in\mathbb{R}^+$ s.t. a.s. $\lim_{t\rightarrow\infty}f_{i,t}$ exists and is equal to $f_i$,
\item a.s. $\lim_{t\rightarrow\infty}\sum_{j=1}^\infty f_{j,t}j^\gamma$ exists, is finite, and is equal to $\sum_{j=1}^\infty f_jj^\gamma$.
\end{enumerate}
The limit of the fraction of bins that contain $i$ balls (denoted by $f_i$) satisfies
\[ f_i \propto i^{-\left(1+1/\left(1-\bar{p}\right)\right)}. \]

\subsection*{Mapping from retweet graph to balls and bins}
In this section, we show that every retweet graph $G_t$ can be mapped to a state of the Polya process, with $\gamma=1$ and $\bar{p}=\frac{\lambda}{\lambda+1}$.

\begin{lem}
Every retweet graph $G_t$ can be represented as a state $S$ of the Polya process.
\label{lem:equivalence}\end{lem}
\begin{proof}
\setlength{\unitlength}{1cm}
\begin{figure}[!ht]
\subfloat[$G_t$\label{fig:G_t}]{\centering
\begin{picture}(6,5)
{\color{green} \put(2,4){\circle*{.5}}}\put(2,4){\circle{.5}}
{\color{green} \put(1,3.5){\circle*{.5}}}\put(1,3.5){\circle{.5}}
{\color{green} \put(1,4.5){\circle*{.5}}}\put(1,4.5){\circle{.5}}
{\color{green} \put(3,3.5){\circle*{.5}}}\put(3,3.5){\circle{.5}}
{\color{green} \put(3,4.5){\circle*{.5}}}\put(3,4.5){\circle{.5}}
\put(2.2,3.8){\vector(3,-2){0.5}}\put(1.8,3.8){\vector(-3,-2){0.5}}\put(2.2,4.2){\vector(3,2){0.5}}\put(1.8,4.2){\vector(-3,2){0.5}}

{\color{blue} \put(2.5,1){\circle*{.5}}}\put(2.5,1){\circle{.5}}
{\color{blue} \put(.5,2){\circle*{.5}}}\put(.5,2){\circle{.5}}
{\color{blue} \put(1.5,2){\circle*{.5}}}\put(1.5,2){\circle{.5}}
{\color{blue} \put(1.5,1){\circle*{.5}}}\put(1.5,1){\circle{.5}}
\put(1.2,2){\vector(-1,0){0.4}}\put(1.5,1.3){\vector(0,1){0.4}}\put(1.8,1){\vector(1,0){0.4}}

{\color{red} \put(3.5,2){\circle*{.5}}}\put(3.5,2){\circle{.5}}
{\color{red} \put(4.5,2){\circle*{.5}}}\put(4.5,2){\circle{.5}}
\put(3.8,2){\vector(1,0){0.4}}

{\color{yellow} \put(4.5,4){\circle*{.5}}}\put(4.5,4){\circle{.5}}
{\color{yellow} \put(4.5,3){\circle*{.5}}}\put(4.5,3){\circle{.5}}
\put(4.5,3.7){\vector(0,-1){0.4}}
\end{picture}
}
\subfloat[$S$\label{fig:polya_state}]{\centering
\begin{picture}(6,5)
\put(1.5,3){\line(0,1){2}}\put(3.5,3){\line(0,1){2}}\put(1.5,3){\line(1,0){2}}
{\color{green} \put(2.1,3.6){\circle*{.5}}}\put(2.1,3.6){\circle{.5}}
{\color{green} \put(3,4.3){\circle*{.5}}}\put(3,4.3){\circle{.5}}
{\color{green} \put(3.1,3.6){\circle*{.5}}}\put(3.1,3.6){\circle{.5}}
{\color{green} \put(1.9,4.4){\circle*{.5}}}\put(1.9,4.4){\circle{.5}}
{\color{green} \put(2.5,4.8){\circle*{.5}}}\put(2.5,4.8){\circle{.5}}

\put(4.5,3){\line(0,1){2}}\put(6.5,3){\line(0,1){2}}\put(4.5,3){\line(1,0){2}}
{\color{yellow} \put(5.1,3.6){\circle*{.5}}}\put(5.1,3.6){\circle{.5}}
{\color{yellow} \put(5.8,4.3){\circle*{.5}}}\put(5.8,4.3){\circle{.5}}

\put(4.5,0){\line(0,1){2}}\put(6.5,0){\line(0,1){2}}\put(4.5,0){\line(1,0){2}}
{\color{red} \put(5.1,.6){\circle*{.5}}}\put(5.1,.6){\circle{.5}}
{\color{red} \put(5.8,1.3){\circle*{.5}}}\put(5.8,1.3){\circle{.5}}

\put(1.5,0){\line(0,1){2}}\put(3.5,0){\line(0,1){2}}\put(1.5,0){\line(1,0){2}}
{\color{blue} \put(2.1,.6){\circle*{.5}}}\put(2.1,.6){\circle{.5}}
{\color{blue} \put(2.8,1.3){\circle*{.5}}}\put(2.8,1.3){\circle{.5}}
{\color{blue} \put(2.1,1.6){\circle*{.5}}}\put(2.1,1.6){\circle{.5}}
{\color{blue} \put(2.9,.4){\circle*{.5}}}\put(2.9,.4){\circle{.5}}
\end{picture}
}
\caption{Mapping from retweet graph $G_t$ (a) to Polya process state $S$ (b).\label{fig:mapping}}
\end{figure}
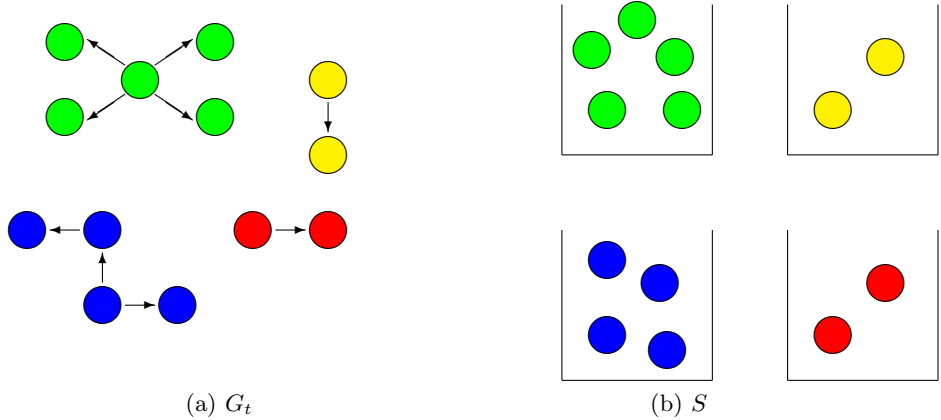
Suppose we have a retweet graph $G_t$, that consists of $k$ components of known sizes, moreover $G_t=\left\{C_1,C_2,\ldots,C_k\mid|C_1|,|C_2|,\ldots,|C_k|\right\}$. For instance, in Figure~\ref{fig:G_t} $G_t=\{C_{green},C_{yellow},C_{blue},C_{red}\mid |C_{green}|=5,|C_{yellow}|=2,|C_{blue}|=4,|C_{red}|=2\}$. First, we take $C_{green}$ that consists of five nodes and fill a bin with five green balls. Then, we take $C_{yellow}$, and fill a bin with $|C_{yellow}|=2$ yellow balls. Next, we take $C_{blue}$ and fill a bin with $|C_{blue}|=4$ blue balls. Finally, we take $C_{red}$ and fill a bin with $|C_{red}=2|$ red balls. These four bins with their corresponding balls then form a state $S$ of the Polya process, depicted in Figure~\ref{fig:polya_state}. Note that by using this procedure for an arbitrary graph $G_t$, we can always construct a state $S$.
\end{proof}

\subsection*{A special case: $p=1$}
In this subsection, we show that the RG model with parameter $p=1$ is equivalent to the Polya proces w.p. $\gamma=1,\bar{p}=\frac{\lambda}{\lambda+1}$. We use this to find the limiting distribution of the component sizes of the RG model.
\begin{thm}
RG model w.p. $p=1$ is equivalent to a Polya process w.p. $\gamma=1,\bar{p}=\frac{\lambda}{\lambda+1}$
\label{thm:equivalence}\end{thm}
\begin{proof}
From Lemma~\ref{lem:equivalence} we know that a retweet graph $G_t$ can be mapped to a state $S$ of the Polya process. Next we show that the probability distribution of an arrival to $G_t$ is identical to the probability distribution of an addition of a ball to the state $S$ in the Polya process, given $p=1$ in the RG model. 

Since $p=1$, we have two types of arivals, $T_1$ and $T_2$. First, we consider a $T_1$ arrival
\begin{equation}
\prob{T_1\mbox{ arrival}} = \frac{\lambda}{\lambda+1} = \bar{p} = \prob{\mbox{new bin is created}}.
\end{equation}
Then, for a $T_2$ arrival, a new node is added to an existing component. The probability that a new node arrives to component $i$ in $G_t$ is
\begin{equation}
\prob{\mbox{arrival to component }i\mbox{ in }G_t} = \frac{|C_i|}{|V_t|}.
\label{eq:rg_prob_dist}\end{equation}
For the Polya process, the probability that a new ball arrives in bin $i$ in $S$ is as follows
\begin{equation}
\prob{\mbox{arrival to bin }i\mbox{ in }S} = \frac{X_i^\gamma}{\sum_{j=1}^kX_j^\gamma} = \frac{X_i}{\sum_{j=1}^kX_j}.
\label{eq:polya_prob_dist}\end{equation}

Using these equations and Lemma~\ref{lem:equivalence}, we find that
\begin{eqnarray*}
\prob{\mbox{arrival to component }i\mbox{ in }G_t} &=& \frac{|C_i|}{|V_t|} = \frac{|C_i|}{\sum_{j=1}^k|C_j|} = \frac{X_i}{\sum_{j=1}^kX_j}, \\
& = & \prob{\mbox{arrival to bin }i\mbox{ in }S}.
\end{eqnarray*}
Thus, the probability distribution of an arrival to the RG model w.p. $p=1$ is identical to the probability distribution of the Polya process w.p. $\gamma=1,\bar{p}=\frac{\lambda}{\lambda+1}$. In combination with Lemma~\ref{lem:equivalence}, we conclude that the RG model is equivalent to a Polya process with parameters $\gamma=1$ and $\bar{p}=\frac{\lambda}{\lambda+1}$.
\end{proof}
From this equivalence, we immediately obtain the limiting component size distribution of the RG model from \cite{chung2003generalizations}, given $p=1$.
\begin{cor}\label{cor:dist_comp}
For $f_i$, the fraction of components of size $i$, it holds that
\[  f_i  \propto i^{-\left(\lambda+2\right)}, \]
for the RG model w.p. $p=1$.
\end{cor}

\section{Component size distribution}\label{sec:theory}
Using the fact that the RG model can be mapped to a Polya process, we derive the limiting behaviour of the component size distribution for general $p$.

\begin{thm}\label{thm:dist_comp}
For the RG model, the limit $f_i$ of the fraction of components of size $i$ a.s. satisfies
\[  f_i  \propto i^{-\left(1+\frac{\lambda+1}{p}\right)}. \]
\end{thm}
\begin{proof}
Given the mapping from the retweet graph to the Polya process, we can derive $f_i$ for the RG model similar to the derivation in \cite{chung2003generalizations}. In that paper, the authors define $p_{i,t}$ as follows
\begin{equation} p_{i,t}:= \prob{\mbox{ball at time }t\mbox{ is placed in bin of size }i},\label{eq:def_p}\end{equation}
with the convention $p_{0,t}=\bar{p}$. Therefore, we can find an expression for $f_i$ by finding an expression for $p_{i,t}$ in the RG model and then following a similar line of reasoning as in \cite{chung2003generalizations}. Note that (\ref{eq:def_p}) is equal to 
\[ p_{i,t}:= \prob{\mbox{ball at time }t\mbox{ increases a bins size to }i+1}.\]
Rewriting this equation to the RG model it holds that, 
\[ p_{i,t} = \prob{\mbox{arrival at time }t\mbox{ results in a component of size }i+1}. \]
Then, let $T_x \rightarrow |C_n|=i+1$ denote a $T_x$ arrival to component $C_n$ augmenting its size to $i+1$ and let $T_3+C_o$ denote that a $T_3$ arrival connects to component $C_0$. We find that for $i\geq 2$
\begin{eqnarray*}
p_{i,t} &=& \prob{T_2\rightarrow|C_n|=i+1}\cdot\prob{T_2\mbox{ arrival}} + \prob{T_3\rightarrow|C_n|=i+1}\cdot\prob{T_3\mbox{ arrival}}, \\
&=& \prob{T_2\rightarrow|C_n|=i+1}\cdot\prob{T_2\mbox{ arrival}}, \\
&& + \sum_{k=1}^i\prob{T_3\rightarrow|C_n|=i+1\mid T_3+C_o,|C_o|=k}\cdot\prob{T_3\mbox{ arrival}}\cdot\prob{T_3+C_o\mid|C_o|=k}, \\
& = & \frac{f_{i,t}\cdot i}{\sum_{j=1}^\infty f_{j,t}\cdot j}\cdot\frac{p}{\lambda+1} + \sum_{k=1}^i \frac{f_{k,t}\cdot k}{\sum_{j=1}^\infty f_{j,t}\cdot j}\cdot\frac{1-p}{\lambda+1}\cdot\frac{2\cdot k\cdot\left(i-k+1\right)}{|V_t|^2-|V_t|}, \\
& = & \frac{f_{i,t}\cdot i}{\sum_{j=1}^\infty f_{j,t}\cdot j}\cdot\frac{p}{\lambda+1} + \sum_{k=1}^i \frac{f_{k,t}\cdot k}{\sum_{j=1}^\infty f_{j,t}\cdot j}\cdot\frac{1-p}{\lambda+1}\cdot\frac{2\cdot k\cdot\left(i-k+1\right)}{\left(\sum_{j=1}^\infty f_{j,t}\cdot j\right)^2-\sum_{j=1}^\infty f_{j,t}\cdot j},
\end{eqnarray*}
with the convention that $p_{0,t}=\prob{T_1}=\frac{\lambda}{\lambda+1}$. 

Then, let $f_{i,t}\rightarrow f_i$ a.s. as $t\rightarrow\infty$. Since the RG model can be mapped to the Polya process with parameter $\gamma=1$ by Lemma~\ref{lem:equivalence}, the limit $\sum_{j=1}^\infty f_j\cdot j$ exists and is equal the average bin size, which we will denote by $\overline{C}$. Thus, the aforementioned assumptions also hold for the RG model. Using this and defining $c=\frac{p_{i-1} - p_i}{f_i}$, we find
\begin{eqnarray*}
c\cdot f_i & = & p_{i-1} - p_i, \\
&=& \frac{f_{i-1}\cdot\left(i-1\right)}{\overline{C}}\cdot\frac{p}{\lambda+1} + \sum_{k=1}^{i-1} \frac{f_k\cdot k}{\overline{C}}\cdot\frac{1-p}{\lambda+1}\cdot\frac{2\cdot k\cdot\left(i-k\right)}{\overline{C}\cdot\left(\overline{C}-1\right)}, \\
&&- \left(\frac{f_i\cdot i}{\overline{C}}\cdot\frac{p}{\lambda+1} + \sum_{k=1}^i \frac{f_k\cdot k}{\overline{C}}\cdot\frac{1-p}{\lambda+1}\cdot\frac{2\cdot k\cdot\left(i-k+1\right)}{\overline{C}\cdot\left(\overline{C}-1\right)}\right), \\
&=& \frac{p}{\lambda+1}\cdot\frac{f_{i-1}\cdot\left(i-1\right)-f_i \cdot i}{\overline{C}} + \frac{1-p}{\lambda+1}\cdot\left(\sum_{k=1}^{i-1} \frac{f_k\cdot k}{\overline{C}}\cdot\frac{2\cdot k\cdot\left(i-k\right)}{\overline{C}\cdot\left(\overline{C}-1\right)} - \sum_{k=1}^i \frac{f_k\cdot k}{\overline{C}}\cdot\frac{2\cdot k\cdot\left(i-k+1\right)}{\overline{C}\cdot\left(\overline{C}-1\right)}\right), \\
&=& \frac{p}{\lambda+1}\cdot\frac{f_{i-1}\cdot\left(i-1\right)-f_i \cdot i}{\overline{C}} - \frac{1-p}{\lambda+1}\cdot\left(\frac{f_i\cdot i}{\overline{C}}\cdot\frac{2\cdot i}{\overline{C}\cdot\left(\overline{C}-1\right)} + \sum_{k=1}^{i-1} \frac{f_k\cdot k}{\overline{C}}\cdot\frac{2\cdot k}{\overline{C}\cdot\left(\overline{C}-1\right)}\right), \\
&\leq& \frac{p}{\lambda+1}\cdot\frac{f_{i-1}\cdot\left(i-1\right)-f_i\cdot i}{\overline{C}}-\frac{1-p}{\lambda+1}\cdot\frac{f_i\cdot i}{\overline{C}}\frac{2\cdot i}{\overline{C}\cdot\left(\overline{C}-1\right)}, \\
&=& \frac{p\cdot f_{i-1}\cdot\left(i-1\right)+\left[\left(\frac{2\cdot i}{\overline{C}\cdot\left(\overline{C}-1\right)}-1\right)\cdot p-\frac{2\cdot i}{\overline{C}\cdot\left(\overline{C}-1\right)}\right]\cdot f_i\cdot i}{\left(\lambda+1\right)\cdot\overline{C}}.
\end{eqnarray*}
And therefore, for $i\geq 2$ it holds that
\[ f_i \leq \frac{p\cdot \left(i-1\right)}{\left(\lambda+1\right)\cdot c\cdot\overline{C}-\left[\left(\frac{2\cdot i}{\overline{C}\cdot\left(\overline{C}-1\right)}-1\right)\cdot p-\frac{2\cdot i}{\overline{C}\cdot\left(\overline{C}-1\right)}\right]\cdot i}\cdot f_{i-1}. \]
Using this expression and defining $f_i\propto g(i)$ as $f_i = c(1+o(1))g(i)$, we find that
\begin{eqnarray} 
f_i &\leq& f_1\cdot\prod_{j=2}^i \frac{p\cdot \left(j-1\right)}{\left(\lambda+1\right)\cdot c\cdot\overline{C}-\left[\left(\frac{2\cdot j}{\overline{C}\cdot\left(\overline{C}-1\right)}-1\right)\cdot p-\frac{2\cdot j}{\overline{C}\cdot\left(\overline{C}-1\right)}\right]\cdot i}, \nonumber\\
&\propto& \prod_{j=2}^i \frac{p\cdot\left(j-1\right)}{\lambda+1+p\cdot j} = \prod_{j=2}^i \frac{j-1}{j+\frac{\lambda+1}{p}} \propto \frac{\Gamma\left(i\right)}{\Gamma\left(i+1+\frac{\lambda+1}{p}\right)} \propto i^{-\left(1+\frac{\lambda+1}{p}\right)}. 
\label{eq:result}\end{eqnarray}
which indicates power-law behaviour. 
\end{proof}

\subsection*{Validation of results}
To validate these results, we ran multiple simulations using our model and plotted the probability density function (pdf) for each of these runs. We compare these simulations to the Yule distribution, with parameters $\alpha$ and $x_{\min}$. Since we are analysing the distribution of component sizes, it follows that $x_{\min}=1$. Then, from Equation~\ref{eq:result}, we find that $\alpha=\frac{\lambda+1}{p}+1$. Let $\Gamma(\cdot)$ denote the Gamma function, by \cite{clauset2009power}, the pdf of this distribution is as follows,
\begin{equation} 
f\left(x\right) = \left(\frac{\lambda+1}{p}\right)\frac{\Gamma\left(\frac{\lambda+1}{p}+1\right)}{\Gamma\left(1\right)}\frac{\Gamma\left(x\right)}{\Gamma\left(x+\frac{\lambda+1}{p}+1\right)}. 
\label{eq:pdf_yule}\end{equation}

In Figure~\ref{fig:simulation_results} we depict the results for several values of $p$ of the simulations with $t=100,000$, $\lambda=\frac{1}{3}$ and $q=0.9$ as values for the other parameters of the RG model. Also depicted in Figure~\ref{fig:simulation_results} as black squares are the values of the pdf of the Yule distribution, shown in Equation~\ref{eq:pdf_yule}. 

In Figure~\ref{fig:simulation_results} we see that for every run, there is one really large component. In \cite{thij2014modelling}, we named these components the Largest Connected Component (LCC) and we mention the fraction of nodes in the LCC in the legend of Figure~\ref{fig:simulation_results}. Note that all the values for the other component sizes are slightly below the Yule distribution. This fact supports our claim that the component size distribution shows power-law behaviour. These values are slightly below the Yule distribution, since Equation~\ref{eq:result} is an upper bound for $f_i$.

\begin{figure}[!ht]\centering
\subfloat[$p=0.4$\label{fig:results_p040}]{
\includegraphics[height=0.4\textwidth]{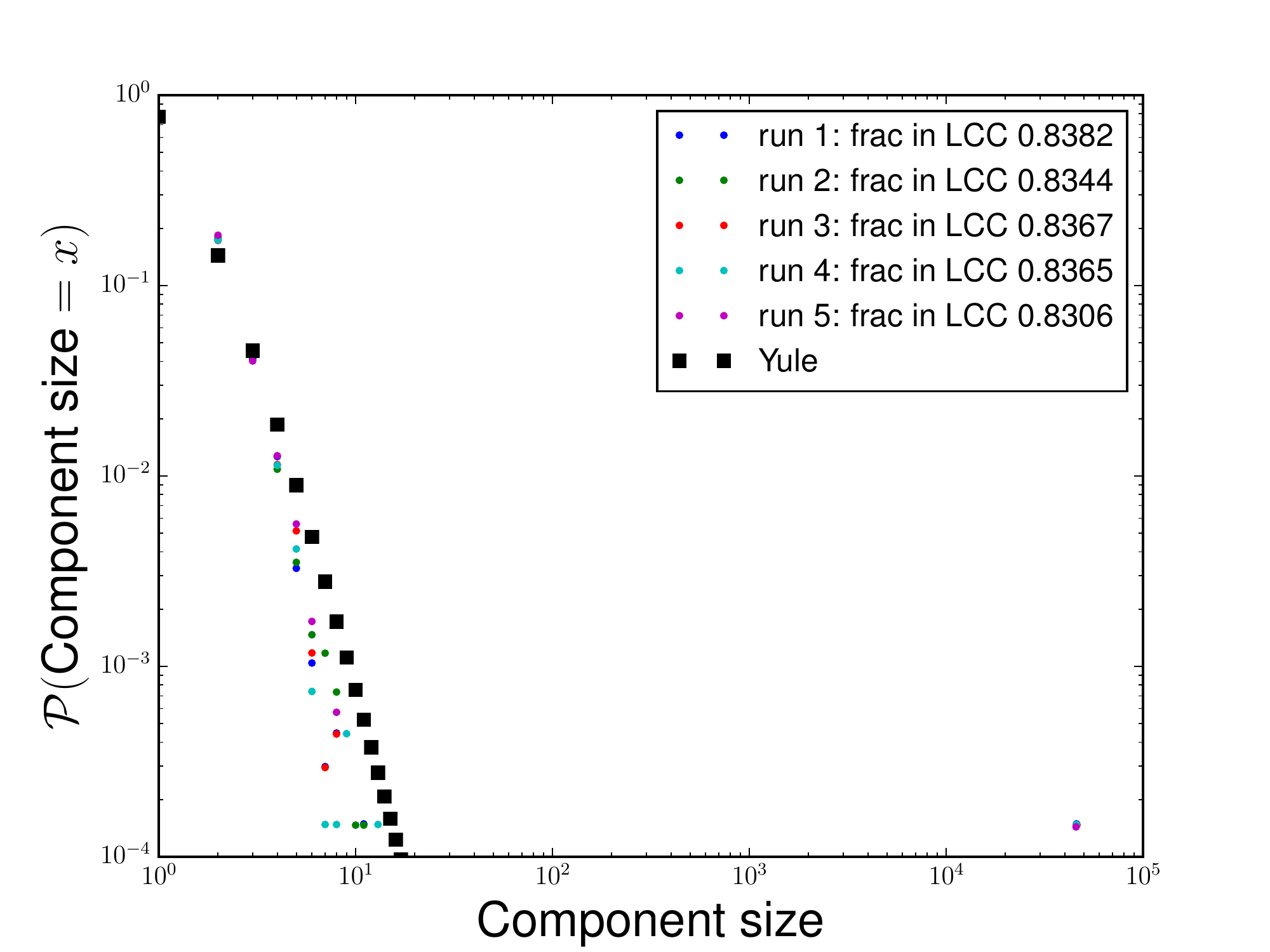}\label{fig:results_p040}}
\subfloat[$p=0.8$\label{fig:results_p080}]{
\includegraphics[height=0.4\textwidth]{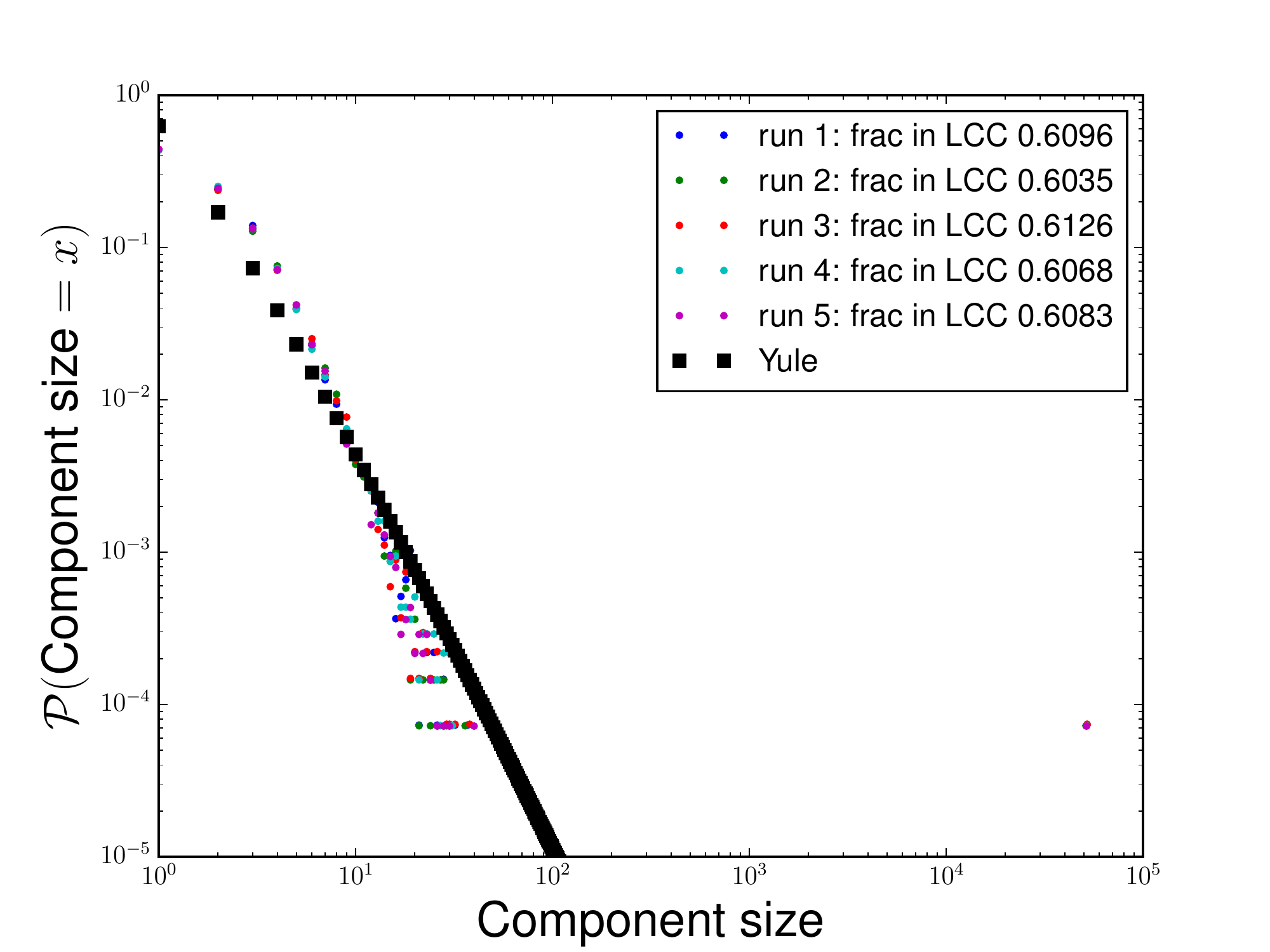}\label{fig:results_p080}}
\caption{Plots of the pdf of the component size distribution. For these simulations we used $t=100,000$, $\lambda=\frac{1}{3}$ and $q=0.9$.\label{fig:simulation_results}}
\end{figure}

\section{Conclusion and Discussion}\label{sec:conclusion}
In this paper, in which we have extended our previous work on the RG model, we derived the limiting behaviour of the component sizes. Through a mapping of the RG model to the Polya process, we found that the distribution of component sizes shows power-law behaviour. Note that Corrolary~\ref{cor:dist_comp} is identical to Theorem~\ref{thm:dist_comp} when we fill in $p=1$. Thus, interestingly enough, the possibility of merging components does not seem to affect this limiting behaviour greatly. The validation in this paper is based on simulated results, thus the model has currently not been tested to \emph{Twitter} datasets. 

Moreover, the model used in this paper can also easily be extended to a less general setting. For instance, the superstar parameter $q$ is assumed to be equal for every message tree. This can be easily extended to an individual superstar parameter per message tree $q_i$. This addition does not change the result of the analysis shown in this paper and therefore is not included here. Another aspect that could be taken into account in future work, is the weigth of an edge. Given that there are multiple retweets between $u$ and $v$, this can be taken into account by adding a weigth to every edge. A last extension that could prove to be interesting, is to use time-varying parameters. For simplicity, this has not been taken into account.

We plan to explore these aspects in our future research.

\bibliographystyle{abbrv}
\bibliography{csizebib}

\flushend
\end{document}